\documentclass[conference]{IEEEtran}
\IEEEoverridecommandlockouts
\usepackage{bbm}
\usepackage{tcolorbox}
\usepackage{cite}
\usepackage{amsmath,amssymb,amsfonts,subfigure}
\usepackage{bm}
\usepackage{algorithmic}
\usepackage{graphicx}
\usepackage{textcomp}
\usepackage{xcolor}
\def\BibTeX{{\rm B\kern-.05em{\sc i\kern-.025em b}\kern-.08em
    T\kern-.1667em\lower.7ex\hbox{E}\kern-.125emX}}

\usepackage{color}
\usepackage{mathtools}
\usepackage{booktabs}
\usepackage{multirow}

\usepackage{comment}
\newtheorem{theorem}{Theorem}
\newtheorem{lemma}{Lemma}

\newtheorem{proposition}{Proposition}

\IEEEoverridecommandlockouts
\allowdisplaybreaks

\begin{document}
\title{Keep, Customize, or Exit: Default Design and Token Pricing in LLM Reasoning Services \\
{\footnotesize 
\thanks{*Ahmet Bugra Gundogan and Yigit Turkmen contributed equally to this work and share first authorship. This work was supported by the TUBITAK 2232-B program (Project No: 124C533).}
}}

\author{
\IEEEauthorblockN{
Ahmet Bugra Gundogan\textsuperscript{*},
Yigit Turkmen\textsuperscript{*},
Melih Bastopcu
}
\IEEEauthorblockA{
Department of Electrical and Electronics Engineering\\
Bilkent University, Ankara, Turkey\\
bugra.gundogan@bilkent.edu.tr,
yigit.turkmen@ug.bilkent.edu.tr,
bastopcu@bilkent.edu.tr
}

}
\maketitle

\begin{abstract}
We study a large language model (LLM) service in which a provider chooses a per-token price and a default reasoning-token allocation, while a user may accept the default, customize the allocation, or exit. Larger allocations can improve accuracy but increase token cost and latency. We model this interaction as a Stackelberg game and derive the user’s unique optimal customized allocation in closed form. For any price, the acceptable defaults form either an empty set or a compact interval. We characterize the provider’s optimal default through a three-regime rule, reduce equilibrium computation to a one-dimensional price optimization, and prove the existence of the equilibrium. We further show that defaults affect the implemented reasoning allocation only when users value the convenience of avoiding customization; otherwise, every service-providing outcome implements the user’s optimal customized allocation. Experiments with two compact open-weight reasoning models on five mathematics and science benchmarks support the accuracy–token model and show how model and task characteristics determine equilibrium prices, defaults, and reasoning allocations.

\end{abstract}


\section{Introduction}
Large language models (LLMs) are increasingly offered through services in which users pay for token consumption and experience latency that depends on the amount of inference-time computation. Many reasoning models additionally allow the reasoning-token budgets to be adjusted, either explicitly through a token allocation or implicitly through service configurations. This creates a joint service-design problem involving price, response quality, latency, and the default reasoning allocation presented to the user. 

Recent studies on test-time scaling demonstrate that allocating
additional inference-time computation can substantially improve
reasoning performance, although the gains depend on the model, task,
and inference budget~\cite{snell2025scaling,wu2025inferencescalinglawsempirical}.
At the same time, budget-forcing and token-budget-aware methods show
that the length of the reasoning process can be explicitly controlled
rather than treated as an incidental property of generation
\cite{muennighoff-etal-2025-s1,han-etal-2025-token}. These developments make the
reasoning-token allocation a natural service parameter that can be exposed
and configured by an LLM provider.

The availability of compact open-weight reasoning models makes this
service-design problem particularly relevant. Since the service provider
controls the inference stack, it can directly enforce reasoning-token
allocations, measure model-specific generation latency, and expose
configurable reasoning allocations to users. This setting is representative
of self-hosted and resource-conscious LLM services, including private,
on-premises, and accelerator-equipped edge-server deployments.

Existing test-time compute studies typically treat the reasoning allocation
as an externally specified constraint or as a quantity selected by an
algorithmic controller. In our setting, however, the provider must
also determine how the configurable reasoning resource is priced and
which allocation is presented to the user by default. The default is
strategically distinct from the price: the price determines the marginal
cost of additional reasoning, whereas the default determines the initial
service configuration encountered by the user. Furthermore, prior work has shown
that defaults can influence decisions through convenience, inertia, and
status-quo effects
\cite{samuelson,powerofsuggestion,carroll2009optimal}.

In this work, we consider a system consisting of an LLM service provider, hereafter referred to as the LLM provider, and a representative user. The LLM provider offers a configurable reasoning-token allocation and a per-token price to the user. After observing the offer, the user may retain the default, customize the allocation, or exit the service. 
More broadly, test-time scaling enables configurable
\emph{reasoning-as-a-service}, in which inference-time computation is
exposed as an adjustable and metered service resource
\cite{liu2026mores}. This emerging paradigm
requires principled mechanisms for pricing reasoning and selecting the
default allocation presented to users. 

\subsection{Related Work}
\vspace{-0.1cm}

\subsubsection{Test-Time Compute and Reasoning Token Allocations}
Existing approaches elicit and aggregate intermediate
reasoning paths through chain-of-thought prompting, self-consistency,
verification, and search
\cite{wang2023selfconsistency,wei2022chain,cobbe2021trainingverifierssolvemath,lightman2024lets,yao2023tree}.  Previous work also
emphasizes evaluating reasoning methods under matched token allocations
\cite{wang-etal-2024-reasoning-token}. Furthermore, 
other approaches study
adaptive allocation, constrained reasoning, and anytime inference \cite{sun-etal-2025-empirical,lin2026plan,
zhang-etal-2026-budget}. 
These works generally treat the inference
allocation as an external constraint or as a quantity selected by an
algorithmic or heuristic controller. In contrast, we model the reasoning-token
allocation as an internal service-design variable. The provider jointly
selects a per-token price and a default reasoning allocation while
anticipating whether the user will retain the default, customize the
allocation, or exit the service.

\subsubsection{Pricing and Latency-Aware LLM Serving} 
Cost-efficient LLM deployment has been extensively studied through
model routing, cascading, and ensemble selection. FrugalGPT,
RouteLLM, HybridLLM, GraphRouter, and MixLLM route queries among
models with different quality, monetary cost, and latency
characteristics
\cite{chen2024frugalgpt,ong2025routellm,ding2024hybrid,
feng2025graphroutergraphbasedrouterllm,balancing_yigit,wang-etal-2025-mixllm, turkmen2026don}. Additionally, RouterBench provides a benchmark
for such routing methods \cite{hu2024routerbench}. These approaches make horizontal
choices among models, agents, or ensembles, whereas our framework
controls service quality in a different way by varying the reasoning depth of a given LLM.

Moreover, Stackelberg games have been widely used for service pricing in
cloud, edge, IoT, and computation-offloading markets \cite{cardellini2026, dynamic_trust, dING2023128429, tutuncuoglu, saxena}. More closely related studies apply
Stackelberg pricing to large-model rental and competitive LLM services, where providers set prices and users select among available
models or platforms \cite{wu2024lmaas, Guo_Bai_Jin_2026}. Accordingly, our focus is not model selection or network routing, but the joint design of pricing and default reasoning allocation within an LLM service.

The closest
related works are  \cite{Guo_Bai_Jin_2026} 
and \cite{velasco2026testtimecomputegames}. Guo \emph{et
al.}~\cite{Guo_Bai_Jin_2026} study provider pricing in a competitive LLM
routing market and develop a data-calibrated learning method for
solving the resulting mathematical program with equilibrium
constraints. Velasco \emph{et al.}~\cite{velasco2026testtimecomputegames}
study competition among providers over test-time compute. They show
that the pay-per-compute equilibrium can be socially inefficient and
propose a reverse second-price auction as a remedy, focusing on welfare
and mechanism design in a normal-form game.  In contrast, we consider a single self-hosted model whose
reasoning depth determines service quality. We jointly optimize the
per-token price and a user-overridable default reasoning allocation when
the user may keep the default, customize the allocation, or exit. To the
best of our knowledge, neither prior work models a default reasoning
allocation or a default-specific convenience benefit. Consequently,
neither contains an analogue of our default-acceptance region or our
characterization of when the default has independent allocative power.

\subsection{Main Contributions}
\vspace{-0.15cm}
Our main contributions are as follows:
\begin{itemize}
    \item Unlike prior test-time compute studies that treat the
    reasoning allocation as an externally specified constraint, we formulate
    the default reasoning allocation as an endogenous service-design
    variable in a provider--user Stackelberg game.

    \item We derive the user's unique customized reasoning allocation and
    characterize the complete default-acceptance region, including its
    feasibility condition and closed-form boundaries.

    \item We characterize the provider's optimal default through a
    three-regime solution and reduce the equilibrium computation to a
    one-dimensional price optimization. We also show that the default
    has allocative power only in the presence of a positive convenience
    benefit.
    \item We establish the existence of a Stackelberg equilibrium and reduce the
service-provision decision to a scalar comparison: the provider
serves if and only if the optimized service value is nonnegative.
    \item Through our experiments using two compact open-weight reasoning
    models across five benchmarks, we fit and assess the accuracy--token model
    and demonstrate how model and task characteristics affect
    equilibrium prices and reasoning allocations.

\end{itemize}

\section{System Model and Problem Formulation}
\label{sec:system_model}
\begin{figure}[tb]
  \centering
\includegraphics[width=0.99\columnwidth]{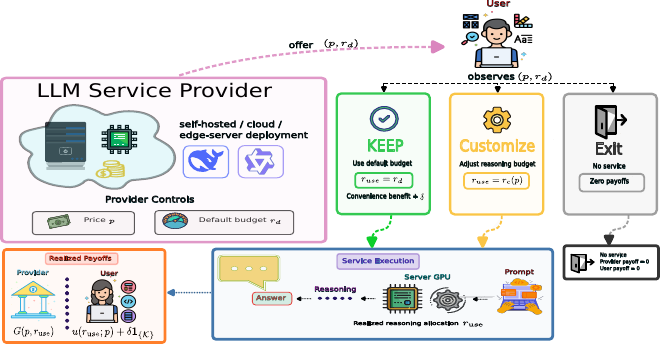}
  \caption{Interaction between an LLM service provider and a representative user. The provider selects the per-token price $p$ and default reasoning allocation $r_d$, after which the user keeps the default, customizes the allocation, or exits.}
  \label{fig:system}
  \vspace{-0.15cm}
\end{figure}

We consider a system where an LLM service provider offers a configurable reasoning-token allocation and a per-token price to a representative user, as illustrated in Fig.~\ref{fig:system}. The user submits a task from a fixed task class, while the provider operates an LLM whose reasoning-token allocation can be configured on a per-request basis.\footnote{We assume that the LLM provider processes at most one task at a time, and any request arriving while another task is being processed is rejected rather than stored in a queue.} The provider first commits to a per-token price $p \geq 0$ and a default reasoning allocation $r_d \geq 0$. After observing $(p, r_d)$, the user either keeps the default, selects a customized reasoning allocation, or exits the service.
We adopt a complete-information formulation in which all
model parameters and payoff functions are common knowledge. This
representative-user model isolates the strategic role of the default;
user and task heterogeneity are left for subsequent analysis.

\subsection{System Model and Utilities}
\label{subsec:system_model_utilities}

Let $r\geq0$ denote the reasoning-token allocation enforced by the
provider for a service request, where $r=0$ corresponds to the baseline
response without additional reasoning. The allocation is binding: the
provider ensures that $r$ reasoning tokens are generated before the
final answer is produced. Such allocations can be implemented using
decoding-time budget forcing~\cite{muennighoff-etal-2025-s1}. Thus,
$r$ is the realized reasoning-token allocation rather than an expected
token count. We treat $r$ as continuous for analytical tractability,
while the experiments use discrete enforced allocations. Given a reasoning allocation $r$, we model the probability of a correct
response $Q(r)$, the expected service latency $t(r\!)$, and the expected
number of billed tokens $T(r\!)$ as
\begin{align}
    Q(r)&=D+A\left(1-e^{-br}\right), \label{eq:accuracy_model}\\
    t(r)&=t_0+cr,\qquad T(r)=T_b+r. \label{eq:latency_billing_models}
\end{align}
Here, $D\geq 0$ denotes the baseline accuracy obtained without allocating any reasoning tokens, while $A>0$ represents the maximum additional accuracy gain attainable through reasoning. The parameter $b>0$ controls how quickly the gain saturates, thereby capturing diminishing returns from allocating additional reasoning tokens. Moreover, $t_0>0$ is the baseline service latency, $c>0$ is the incremental latency per reasoning token, and $T_b>0$ is the expected number of billed non-reasoning tokens,
comprising the input prompt and the final-answer segment.\footnote{We use the term \emph{expected} because, although the
reasoning allocation is fixed to $r$, the base latency and the number
of billed input and output tokens may vary from one task to another; $t_0$ and
$T_b$ denote their respective means.}
Thus, under the models in \eqref{eq:accuracy_model} and \eqref{eq:latency_billing_models}, additional reasoning improves accuracy at a diminishing rate while increasing latency and the number of billed tokens linearly. We abstract the details of the underlying computational infrastructure through the latency model $t(r)$ and focus on the strategic interaction between the provider’s pricing and default-configuration decisions and the user’s service and reasoning choices. Finally, we impose $D+A\leq 1$, which guarantees that $Q(r)\in[0,1]$ for all $r\geq 0$. We denote the LLM provider’s price per token by $p\geq 0$.

For a given \(p\) and \(r\), the user's baseline expected utility is
\begin{align}
    u_0(r,p)=vQ(r)-pT(r)-\theta t(r),
    \label{eq:user_baseline_utility}
\end{align}
where $v\!>\!0$ is the user's value of a correct response and $\theta>0$ is the user's
latency sensitivity. All utility terms are expressed in common
monetary-equivalent units. If the user keeps the default, the user’s utility is
\(U_{\mathcal K}(p,r_d)=u_0(r_d,p)+\delta\), where
\(\delta\geq0\) is a default-specific convenience benefit. It captures
the reduced cognitive and interaction costs of accepting a
preconfigured option and the status-quo advantage associated with
defaults \cite{samuelson,powerofsuggestion,carroll2009optimal}.
Because $\delta$ is associated specifically with retaining the default configuration, it enters both the user’s comparison between the default and customization and the user’s participation decision. Thus, $\delta$ should not be interpreted as a decision-making cost incurred only under customization.

If the user customizes the reasoning tokens, it solves
\begin{align}
    B(p)=\max_{r\geq0}u_0(r,p),\qquad
    r_c(p)=\arg\max_{r\geq0}u_0(r,p),
    \label{eq:customization_problem}
\end{align}
where the maximizer will be shown to be unique in Lemma~\ref{lem:user_customization}.
The customization utility is therefore
$U_{\mathcal C}(p,r_d)=B(p)$, which is not a function of $r_d$. The exit utility is normalized to $U_{\mathcal E}(p,r_d)=0$.

When service is provided with reasoning allocation \(r\), the
provider's expected payoff is
\begin{align}
    G(p,r)=(p-\rho)T(r)+\alpha Q(r)-\beta t(r),
    \label{eq:provider_payoff}
\end{align}
where $\rho\geq0$ is the marginal cost per billed token, $\alpha>0$ is the
weight the provider places on response accuracy, and $\beta\geq 0$ is its
per-unit latency cost. The $\alpha Q(r)$ term represents effects such as user retention,
reputation, and service-level performance, while the latency term
captures resource occupation and delay-related operating costs. If
service is not provided, the provider's payoff is zero. Having defined
both players' payoffs, we next formalize their interaction as a
sequential game in which the provider moves first by selecting the
price and the default reasoning allocation.

\subsection{Stackelberg Game and Problem Formulation}
\label{subsec:stackelberg_game}

We formulate the interaction as a Stackelberg game in which the provider is the
leader and the user is the follower. The provider moves first and
either selects \((p,r_d)\in\mathbb R_{\geq0}^2\) or chooses the
no-service action \(\mathcal N\), under which both players receive
zero payoff. After observing \((p,r_d)\), the user selects an action
\(y\in\mathcal Y=\{\mathcal K,\mathcal C,\mathcal E\}\),
corresponding to keeping the default, customizing, and exiting.
\subsubsection{Provider's Payoff and Problem}
For an offer $(p,r_d)\in\mathbb R_{\geq0}^2$ and a user action
$y\in\mathcal Y$, the provider's payoff is
\begin{align}
\Pi(p,r_d,y)=
\begin{cases}
G(p,r_d), & y=\mathcal K,\\
G\bigl(p,r_c(p)\bigr), & y=\mathcal C,\\
0, & y=\mathcal E.
\end{cases}
\label{eq:provider_payoff_profile}
\end{align}
Let $\mathcal X=\{\mathcal N\}\cup\mathbb R_{\geq0}^2$ denote the
provider's action set, and let $BR(p,r_d)\in\mathcal Y$ denote the
user's best response, specified in \eqref{eq:user_best_response}
below. The provider's induced payoff $J:\mathcal X\to\mathbb R$ is
\begin{align}
J(x)=
\begin{cases}
\Pi\bigl(p,r_d,BR(p,r_d)\bigr), & x=(p,r_d)\in\mathbb R_{\geq0}^2,\\
0, & x=\mathcal N.
\end{cases}
\label{eq:induced_payoff}
\end{align}
An action $x^\star\in\mathcal X$ is a Stackelberg solution if
\begin{align}
J(x^\star)\geq J(x),
\qquad \forall\, x\in\mathcal X.
\label{eq:stackelberg_solution}
\end{align}
If $x^\star=(p^\star,r_d^\star)$, the induced user action is
$BR(p^\star,r_d^\star)$; if $x^\star=\mathcal N$, service is not
offered.

\subsubsection{User's Best Response}
Given an offer \((p,r_d)\), the user compares the keep utility
\(U_{\mathcal K}(p,r_d)\), the customization utility \(B(p)\), and the
exit utility \(0\). To make the best response single-valued on
indifference boundaries, we impose the tie-breaking rule
\(\mathcal K\succ\mathcal C\succ\mathcal E\): the user keeps the
default whenever it is utility-maximizing and, when keeping is not
optimal, customizes rather than exits if both yield zero utility.
Under this rule,
\begin{align}
BR(p,r_d)=
\begin{cases}
\mathcal K,
& U_{\mathcal K}(p,r_d)\geq B(p),\
  U_{\mathcal K}(p,r_d)\geq0,\\
\mathcal C,
& B(p)>U_{\mathcal K}(p,r_d),\
  B(p)\geq0,\\
\mathcal E,
& U_{\mathcal K}(p,r_d)<0,\
  B(p)<0,
\end{cases}
\label{eq:user_best_response}
\end{align}
where the three cases are mutually exclusive and exhaustive. Because
an optimal provider decision may lie on an indifference boundary, the
equilibrium characterization is conditional on this tie-breaking rule;
alternative boundary behavior is discussed after the equilibrium
analysis.

In the following section, we provide the equilibrium analysis for the Stackelberg game formulated in (\ref{eq:stackelberg_solution}) and (\ref{eq:user_best_response}).

\section{Equilibrium Analysis}
\label{sec:equilibrium_analysis}
We find an equilibrium of the game by backward induction: we first characterize the
user's customized allocation \(r_c(p)\) and the best response, then
determine the set of defaults the user accepts, solve the provider's
problem at a fixed price, and finally optimize over the price.
\subsection{User's Best Response}
\label{subsec:user_best_response}

We first characterize the user's optimal reasoning allocation when the
default is rejected. For compactness, we define the auxiliary quantities
\begin{align}
    a\!=\!vA,~\! m(p)\!=\!p+\theta c,~\!
    C(p)\!=\!v(D+A)\!-\!pT_b\!-\!\theta t_0.
    \label{eq:auxiliary_definitions}
\end{align}
Since \(p\geq0\) and \(\theta,c>0\), we have
\(m(p)\geq\theta c>0\) for every feasible price. Substituting
\eqref{eq:accuracy_model} and \eqref{eq:latency_billing_models} into the user's baseline expected utility $u_0(r,p)$ in 
\eqref{eq:user_baseline_utility} gives
\begin{align}
    u_0(r,p)=C(p)-ae^{-br}-m(p)r.
    \label{eq:user_utility_compact}
\end{align}
Next, we show that $u_0(r,p)$ is strictly concave in $r$, which yields
the unique customized allocation $r_c(p)$.

\begin{lemma}
\label{lem:user_customization}
For every \(p\geq0\), the user's objective \(u_0(r,p)\) is strictly
concave in \(r\), and the customization problem in
\eqref{eq:customization_problem} admits a unique solution given by
\begin{align}
r_c(p)=
\begin{cases}
0, & m(p)\geq ab,\\[2mm]
\dfrac{1}{b}\log\!\left(\dfrac{ab}{m(p)}\right),
& m(p)<ab.
\end{cases}
\label{eq:customized_allocation_closed_form}
\end{align}
The corresponding customization value is
\begin{align}
\!\!B(p)\!=\!\!
\begin{cases}
vD-pT_b-\theta t_0,
& m(p)\geq ab,\\[2mm]
C(p)\!-\!\dfrac{m(p)}{b}
\left(1\!+\!\log\left(\dfrac{ab}{m(p)}\right)\!\right),\!\!\!
& m(p)<ab,
\end{cases}\!\!
\label{eq:customization_value_closed_form}
\end{align}
where $\log(\cdot)$ denotes the
natural logarithm. Moreover,
\(B(p)\) is continuously differentiable and strictly decreasing for $p\geq0$, with derivative
\begin{align}
    B'(p)=-\bigl(T_b+r_c(p)\bigr)<0,
\label{eq:customization_value_derivative}
\end{align}
and satisfies \(\lim_{p\to\infty}B(p)=-\infty\).
\end{lemma}
\begin{IEEEproof}
The first and second derivatives of $u_0(r,p)$ in \eqref{eq:user_utility_compact} with respect to $r$ are
\begin{align}
    \frac{\partial u_0(r,\!p)}{\partial r}
  \!  =\!ab e^{\!-br}\!\!-\!m(p),\!\!\!\!\!\!\!\!\qquad
    \frac{\partial^2 u_0(r,\!p)}{\partial r^2}
   \! =\!-ab^2\!e^{-\!br}\!<\!0.
    \label{eq:user_utility_derivatives}
\end{align}
Hence, $u_0(r,p)$ is strictly concave in \(r\). Moreover, $m(p)>0$ implies that $u_0(r,p)$ tends to $-\infty$ as
$r$ grows unboundedly; hence the maximum over $r\geq0$ is
attained, and strict concavity implies it is unique.

If \(m(p)\geq ab\), then we have $\left.
\frac{\partial u_0(r,p)}{\partial r}
\right|_{r=0}
=ab-m(p)\leq0.$ In this case, since the derivative is strictly decreasing in $r$, the unique
maximizer is $r_c(p)=0$. This includes the equality case
$m(p)=ab$.

If $m(p)<ab$, the derivative is positive at $r=0$ and converges to
$-m(p)<0$ as $r$ grows large. The unique interior maximizer therefore
satisfies $ab e^{-br_c(p)}=m(p),$ which yields the second case in
\eqref{eq:customized_allocation_closed_form}. Substituting the two
possible values of $r_c(p)$ into $u_0(r,p)$ gives
\eqref{eq:customization_value_closed_form}.

It remains to verify the claimed properties of $B(p)$. Within either
regime, direct differentiation of
\eqref{eq:customization_value_closed_form} gives $B'(p)=-\left(T_b+r_c(p)\right).$ At the junction $m(p)=ab$, the interior expression
$\tfrac{1}{b}\log(ab/m(p))$ equals zero, so the two branches of
\eqref{eq:customized_allocation_closed_form} coincide and $r_c(p)$ is
continuous. Substituting $m(p)=ab$ into the interior expression for
$B(p)$ yields
\[
    C(p)-\frac{ab}{b}
    =
    vD-pT_b-\theta t_0,
\]
which matches the corner branch
in \eqref{eq:customization_value_closed_form}. Likewise, the interior
derivative $-(T_b+r_c(p))$ equals $-T_b$ at the junction, since
$r_c(p)=0$ there, matching the derivative in the corner regime.
Thus, $B(p)$ is continuously differentiable across the junction.
Since $T_b>0$ and $r_c(p)\geq0$,
\eqref{eq:customization_value_derivative} implies that $B(p)$ is
strictly decreasing. Finally, since $m(p)=p+\theta c$ is increasing
and unbounded in $p$, we have $m(p)\geq ab$ for all sufficiently
large $p$, in which case $B(p)=vD-pT_b-\theta t_0$ decreases
without bound as $p$ grows large.
\end{IEEEproof}

Lemma~\ref{lem:user_customization} shows that the user selects a
positive reasoning allocation if and only if
\begin{align}
    p+\theta c<vAb.
    \label{eq:positive_customization_condition}
\end{align}
Thus, additional reasoning is selected precisely when the marginal
accuracy value of the first reasoning token, $vAb$, exceeds its
effective marginal cost, $p+\theta c$, comprising the price and the
per-token latency cost. In particular, if $\theta c\geq vAb$,
then $r_c(p)=0$ for every $p\geq0$. Substituting
\eqref{eq:customized_allocation_closed_form} and
\eqref{eq:customization_value_closed_form} into
\eqref{eq:user_best_response} completes the characterization of the
user's best response for every given pair $(p,r_d)$.

\subsection{Default Acceptance Region}
\label{subsec:default_acceptance_region}
We next characterize the default reasoning allocations that the
provider can induce at a given price. Under the tie-breaking rule in
\eqref{eq:user_best_response}, the user keeps a default $r_d$ if and
only if keeping yields at least as much utility as both customization
and exit. Accordingly, we define the \emph{acceptance region}
\begin{align}
\mathcal D(p)
=
\left\{
r_d\geq0:
U_{\mathcal K}(p,r_d)\geq B(p),\
U_{\mathcal K}(p,r_d)\geq0
\right\},
\label{eq:default_acceptance_region}
\end{align}
so that $BR(p,r_d)=\mathcal K$ if and only if $r_d\in\mathcal D(p)$. Using $U_{\mathcal K}(p,r_d)=u_0(r_d;p)+\delta$, this set can be
written as
\begin{align}
\mathcal D(p)
=
\left\{r_d\geq0: u_0(r_d;p)\geq h(p)\right\},
\label{eq:acceptance_threshold}
\end{align}
where $h(p)=\max\left\{B(p)-\delta,\,-\delta\right\}.$ When $B(p)\geq0$, the customization constraint determines the
threshold, and $h(p)=B(p)-\delta$. When $B(p)<0$, customization is
dominated by exit, so the participation constraint determines the
threshold, giving $h(p)=-\delta$. Using
\eqref{eq:auxiliary_definitions} and \eqref{eq:user_utility_compact}, we
define
\begin{align}
\!K(p)\!=\!C(p)-h(p),
\quad
z(p)\!=\!
-\frac{ab}{m(p)}
e^{-\frac{bK(p)}{m(p)}},\!\!
\label{eq:K_z_definitions}
\end{align}
which will be used in the following lemma. 
\begin{lemma}
\label{lem:default_acceptance_region}
For every $p\geq0$, the acceptance region $\mathcal D(p)$ is nonempty
if and only if
\begin{align}
    B(p)+\delta\geq0.
    \label{eq:acceptance_nonempty_condition}
\end{align}
When \eqref{eq:acceptance_nonempty_condition} holds, we have
$z(p)\in[-1/e,0)$, and $\mathcal D(p)$ is a compact interval
(possibly a singleton) given by
\begin{align}
\mathcal D(p)
=
\begin{cases}
[0,\overline r(p)], & K(p)\geq a,\\[1mm]
[\underline r(p),\overline r(p)], & K(p)<a,
\end{cases}
\label{eq:acceptance_interval}
\end{align}
where $\overline r(p)
=
\frac{K(p)}{m(p)}
+\frac{1}{b}W_0\bigl(z(p)\bigr)$ and $\underline r(p)=
\frac{K(p)}{m(p)}
+\frac{1}{b}W_{-1}\bigl(z(p)\bigr)$. 
Here, $W_0$ and $W_{-1}$ denote the principal and lower real branches
of the Lambert $W$ function, respectively.
\end{lemma}

\begin{IEEEproof}
By \eqref{eq:acceptance_threshold}, $\mathcal D(p)$ is the
super-level set $\{r_d\geq0: u_0(r_d,p)\geq h(p)\}$. By
Lemma~\ref{lem:user_customization}, $u_0(\cdot,p)$ is strictly
concave and decreases without bound as $r$ grows large; hence every
nonempty super-level set over $r\geq0$ is a compact interval,
possibly a singleton, and $\mathcal D(p)$ is nonempty if and only if
$B(p)\geq h(p)$. Since $\delta\geq0$, the inequality
$B(p)\geq B(p)-\delta$ always holds, while $B(p)\geq-\delta$ is
equivalent to \eqref{eq:acceptance_nonempty_condition}.

By \eqref{eq:user_utility_compact}, the boundary equation
$u_0(r;p)=h(p)$ becomes
\[
    ae^{-br}+m(p)r=K(p).
\]
Substituting $y=K(p)-m(p)r$, which equals $ae^{-br}>0$ on the
boundary, yields
$-\tfrac{by}{m(p)}e^{-\tfrac{by}{m(p)}}=z(p)$, and the
two real branches of the Lambert $W$ function give
\[
    r=\frac{K(p)}{m(p)}+\frac{1}{b}W_k\bigl(z(p)\bigr),
    \qquad k\in\{0,-1\}.
\]
These roots are real whenever $\mathcal D(p)$ is nonempty. If
$ab>m(p)$, then $B(p)$ is given by the interior branch of
\eqref{eq:customization_value_closed_form}, and $B(p)\geq h(p)$
becomes $K(p)\geq\tfrac{m(p)}{b}\bigl[1+\log(ab/m(p))\bigr]$, which
is equivalent to $z(p)\geq-1/e$. If $ab\leq m(p)$, then
$u_0(\cdot,p)$ is nonincreasing on $r\geq0$, nonemptiness reduces to
$u_0(0,p)\geq h(p)$, i.e., $K(p)\geq a$, and setting
$q=ab/m(p)\leq1$ gives $z(p)\geq-qe^{-q}\geq-1/e$, since $qe^{-q}$
attains its maximum $1/e$ at $q=1$. Since $z(p)<0$ for every $p$, both branches are well defined.

Finally, $u_0(0,p)\geq h(p)$ is equivalent to $K(p)\geq a$. In this
case, the super-level set contains $r=0$, so its intersection with
$r\geq0$ is $[0,\overline r(p)]$ with the upper endpoint given by the
larger root $W_0$, proving the first case of
\eqref{eq:acceptance_interval}. If $K(p)<a$, then $r=0$ is rejected
while nonemptiness guarantees interior points with
$u_0(r;p)\geq h(p)$; by concavity, both endpoints are positive roots
of the boundary equation, and since $W_{-1}(z)\leq W_0(z)$, the lower
and upper endpoints are generated by $W_{-1}$ and $W_0$,
respectively, proving the second case. The endpoints coincide
exactly when $z(p)=-1/e$.
\end{IEEEproof}

Two boundary implications will be useful in the subsequent analysis.
First, if $\delta=0$, then $h(p)=B(p)$, and nonemptiness in
\eqref{eq:acceptance_nonempty_condition} requires $B(p)\geq0$. Since
the value $B(p)$ is attained uniquely at $r_c(p)$, we have $\mathcal D(p)=\{r_c(p)\}$ when $\delta=0.$ Thus, without a default convenience benefit, the provider cannot
induce any reasoning allocation other than the user's customized
optimum. Second, at any price satisfying $B(p)+\delta=0$, the
participation and customization constraints bind simultaneously,
$h(p)=B(p)$, and $\mathcal D(p)=\{r_c(p)\}$ holds again.
When $r_c(p)>0$, this collapse corresponds to $z(p)=-1/e$, where the
two Lambert branches coincide; when $r_c(p)=0$, it corresponds to the
boundary $K(p)=a$. Conversely, if $\delta>0$ and $B(p)+\delta>0$,
then $h(p)<B(p)$, and the acceptance region is an interval of
positive length. This is the regime in which the provider can use the
default to steer the user away from $r_c(p)$.

\subsection{Provider's Optimal Default Reasoning for a Fixed Price}
\label{subsec:fixed_price_default}
We next determine the provider's optimal default at a fixed price.
First, we show that it suffices to consider defaults in the
acceptance region. Suppose an offer $(p,r_d)$ induces the user to
customize. By \eqref{eq:user_best_response}, customization requires
$B(p)\geq0$. If the provider instead offers the default
$r_d=r_c(p)$, then
\[
U_{\mathcal K}(p,r_c(p))
=
u_0(r_c(p),p)+\delta
=
B(p)+\delta
\geq B(p)\geq0,
\]
so $r_c(p)\in\mathcal D(p)$ and the user keeps this default under the
tie-breaking rule. The implemented allocation and provider payoff
remain $r_c(p)$ and $G(p,r_c(p))$, respectively; hence every
customization outcome is replicated by an accepted default.
Similarly, an offer that induces exit yields the provider payoff
zero, which is also attained by the no-service action $\mathcal N$.
Therefore, in maximizing $\Pi$, it is without loss of optimality to
restrict attention to $r_d\in\mathcal D(p)$ whenever
$\mathcal D(p)\neq\emptyset$, and to $\mathcal N$ otherwise.

Consequently, for any price $p$ with $\mathcal D(p)\neq\emptyset$,
the provider's optimal service-providing default solves
\begin{align}
    r_d^\dagger(p)
    \in
    \arg\max_{r_d\in\mathcal D(p)}G(p,r_d).
    \label{eq:fixed_price_provider_problem}
\end{align}
By Lemma~\ref{lem:default_acceptance_region}, we write
$\mathcal D(p)=[\underline r(p),\overline r(p)]$, with the convention
$\underline r(p)=0$ when $K(p)\geq a$. Substituting
\eqref{eq:accuracy_model} and \eqref{eq:latency_billing_models} into
\eqref{eq:provider_payoff} gives
\begin{align}
\!\!G(p,r)
\!=\!
(p\!-\!\rho)T_b\!+\!\alpha D\!-\!\beta t_0
\!+\!\gamma(p)\,r+\alpha A\left(1\!-\!e^{-br}\right),\!\!
\label{eq:provider_payoff_compact}
\end{align}
where
\begin{align}
    \gamma(p)=p-\rho-\beta c
    \label{eq:provider_marginal_coefficient}
\end{align}
is the provider's net marginal revenue per reasoning token, after
accounting for the token-generation cost $\rho$ and the per-token
latency cost $\beta c$. In the following proposition, we characterize the provider's optimal reasoning allocation for a given $p$.

\begin{proposition}
\label{prop:fixed_price_default}
For every $p\geq0$ such that
$\mathcal D(p)=[\underline r(p),\overline r(p)]\neq\emptyset$, the
fixed-price problem in \eqref{eq:fixed_price_provider_problem} has a
unique solution given by
\begin{align}
r_d^\dagger(p)
=
\begin{cases}
\overline r(p),
& \gamma(p)\geq0,\\[2mm]
\underline r(p),
& \gamma(p)\leq-\alpha Ab,\\[2mm]
\operatorname{Proj}_{\mathcal D(p)}
\bigl(\widehat r(p)\bigr),
& -\alpha Ab<\gamma(p)<0,
\end{cases}
\label{eq:fixed_price_optimal_default}
\end{align}
where
\begin{align}
    \widehat r(p)
    =
    \frac{1}{b}
    \log\left(-
    \frac{\alpha Ab}{\gamma(p)}
    \right)
    \label{eq:provider_unconstrained_allocation}
\end{align}
and
$\operatorname{Proj}_{\mathcal D(p)}(x)
=\min\{\overline r(p),\max\{\underline r(p),x\}\}$
denotes the projection onto $\mathcal D(p)$.
\end{proposition}
\begin{IEEEproof}
For fixed $p$, differentiating $G(p,r)$ in  \eqref{eq:provider_payoff_compact}
with respect to $r$ gives $\frac{\partial G(p,r)}{\partial r} = \gamma(p)+\alpha Ab e^{-br}$ and $\frac{\partial^2 G(p,r)}{\partial r^2} = -\alpha Ab^2e^{-br}<0$. Thus, $G(p,r)$ is strictly concave in $r$ with strictly decreasing
marginal payoff, and the maximizer over the compact interval
$\mathcal D(p)$ is unique.

If $\gamma(p)\geq0$, then $\frac{\partial G(p,r)}{\partial r}>0$ for every $r\geq0$, so
$G(p,\cdot)$ is strictly increasing on $\mathcal D(p)$ and the
maximizer is $\overline r(p)$.

If $\gamma(p)\leq-\alpha Ab$, then $\left.\frac{\partial G(p,r)}{\partial r}\right|_{r=0}\leq0$, and since
$\frac{\partial G(p,r)}{\partial r}$ is strictly decreasing, $\frac{\partial G(p,r)}{\partial r}<0$ for every $r>0$.
Hence $G(p,\cdot)$ is nonincreasing on $r\geq0$ and strictly
decreasing away from $r=0$, so the maximizer is $\underline r(p)$.

Finally, suppose $-\alpha Ab\!<\!\gamma(p)\!<\!0$. Then $\left.\frac{\partial G(p,r)}{\partial r}\right|_{r=0}\!\!>\!0$, while
$\frac{\partial G(p,r)}{\partial r}$ decreases toward the negative limit $\gamma(p)$ as $r$
grows large. The unconstrained maximizer is therefore the unique
solution of $\left.\frac{\partial G(p,r)}{\partial r}\right|_{r=\widehat r(p)} =0$, which gives
\eqref{eq:provider_unconstrained_allocation}. By strict concavity,
the maximizer over the interval $\mathcal D(p)$ is the projection of
$\widehat r(p)$ onto $\mathcal D(p)$, completing the proof.
\end{IEEEproof}

Proposition~\ref{prop:fixed_price_default} identifies three provider
regimes. When $\gamma(p)\geq0$, an additional reasoning token yields
a nonnegative net margin and a strictly positive accuracy
contribution, so the provider selects the largest accepted default.
When $\gamma(p)\leq-\alpha Ab$, the marginal payoff is nonpositive
already at $r=0$, and the provider selects the smallest accepted
default. In the intermediate regime, the provider has a unique
preferred allocation $\widehat r(p)$ and implements the closest
allocation permitted by the acceptance region.

For the price-optimization problem, define the optimal
accepted-default payoff at price $p$ as
\begin{align}
    V(p)
    =
    G\bigl(p,r_d^\dagger(p)\bigr)
    =
    \max_{r_d\in\mathcal D(p)}G(p,r_d),
    \quad
    \mathcal D(p)\neq\emptyset.
    \label{eq:fixed_price_value}
\end{align}
The value $V(p)$ is the provider's optimal payoff conditional on
serving at price $p$; the comparison with the no-service action is
deferred to the global Stackelberg game problem.

\begin{figure*}[t]
    \centering
    \includegraphics[width=0.20\textwidth]{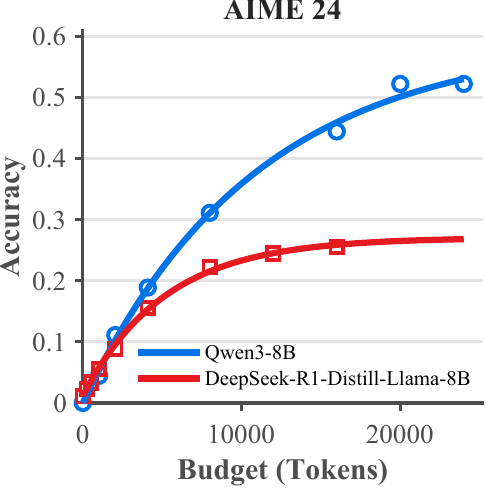}\hfill
    \includegraphics[width=0.20\textwidth]{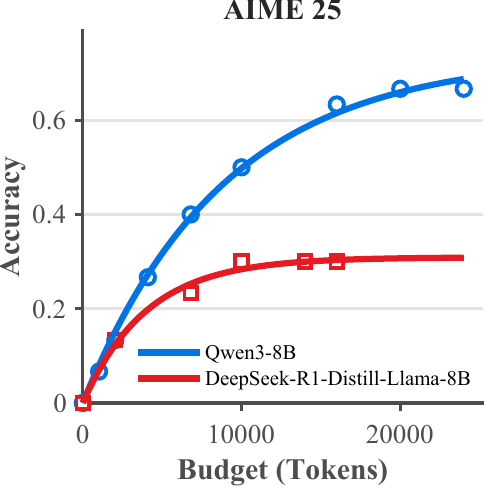}\hfill
    \includegraphics[width=0.20\textwidth]{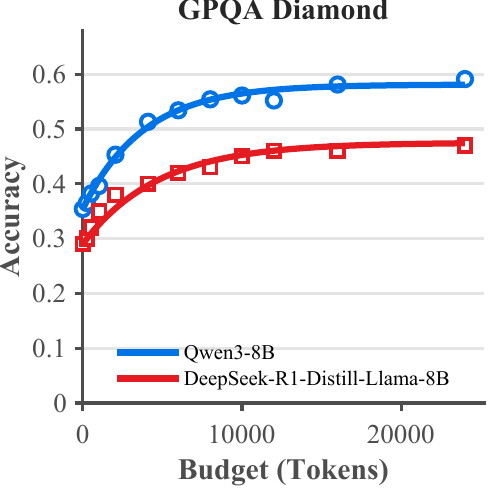}\hfill
    \includegraphics[width=0.20\textwidth]{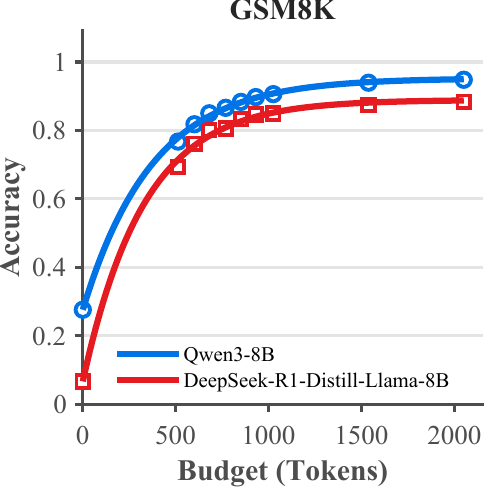}\hfill
    \includegraphics[width=0.20\textwidth]{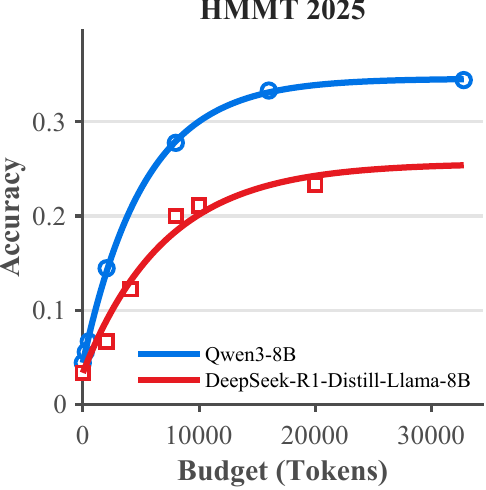}
    
    \vspace{-0.2cm}
    \caption{Performance across five reasoning datasets for Qwen3-8B and DeepSeek-R1-Distill-Llama-8B models. Hollow markers denote empirical dataset evaluations, while solid bold lines represent the theoretical accuracy expression ($Q(r) = D + A(1 - e^{-br})$) fitted to the data.}
    \label{fig:acc_token_plots}
\end{figure*}
\subsection{Price Optimization and Stackelberg Equilibrium}
\label{subsec:price_optimization}
We now optimize the fixed-price value in \eqref{eq:fixed_price_value}
over $p$. By Lemma~\ref{lem:default_acceptance_region}, the set
of prices at which an accepted default can be induced is
\begin{align}
\mathcal P_\delta
=
\{p\geq0:\mathcal D(p)\neq\emptyset\}
=
\{p\geq0:B(p)+\delta\geq0\}.
\label{eq:feasible_price_set}
\end{align}
Whenever $\mathcal P_\delta\neq\emptyset$, define the optimal service
payoff and an associated maximizer as
\begin{align}
V_{\mathrm{serv}}^\star
=
\max_{p\in\mathcal P_\delta}V(p),
\qquad
p^\star\in\arg\max_{p\in\mathcal P_\delta}V(p),
\label{eq:best_service_value}
\end{align}
with $r_d^\star=r_d^\dagger(p^\star)$. The maximizer in
\eqref{eq:best_service_value} need not be unique; $p^\star$ denotes
an arbitrary selection from the set of maximizers, and all subsequent
statements involving $(p^\star,r_d^\star)$ hold for every such
selection. In particular, while the equilibrium \emph{value} is
unique, the equilibrium \emph{price} may be set-valued.
\begin{theorem}
\label{thm:stackelberg_characterization}
If $B(0)+\delta<0$, then $\mathcal P_\delta=\emptyset$ and
$\mathcal N$ is a Stackelberg equilibrium. Otherwise,
$\mathcal P_\delta=[0,\bar p_\delta]$, where $\bar p_\delta$ is the
unique root of $B(p)+\delta=0$ and
$\mathcal D(\bar p_\delta)=\{r_c(\bar p_\delta)\}$. The function
$V(p)$ is continuous on $[0,\bar p_\delta]$, so the maximum in
\eqref{eq:best_service_value} is attained, and the provider's payoff at the
Stackelberg equilibrium is
\begin{align}
    V_{\mathrm{SE}}^\star
    =
    \max\{V_{\mathrm{serv}}^\star,0\},
    \label{eq:global_stackelberg_value}
\end{align}
attained by any $(p^\star,r_d^\star)$ if
$V_{\mathrm{serv}}^\star\geq0$ and by $\mathcal N$ if
$V_{\mathrm{serv}}^\star\leq0$.
\end{theorem}
\begin{IEEEproof}
Suppose $B(0)+\delta<0$. Since $B(p)$ is strictly decreasing
(Lemma~\ref{lem:user_customization}), $B(p)+\delta<0$ for every
$p\geq0$, so $\mathcal P_\delta=\emptyset$. Moreover, for every
$r_d\geq0$,
$U_{\mathcal K}(p,r_d)=u_0(r_d,p)+\delta\leq B(p)+\delta<0$, while
$B(p)<0$; hence the user exits under every offer, every provider
decision yields zero payoff, and $\mathcal N$ is optimal.

Suppose now $B(0)+\delta\geq0$. Since $B(p)$ is continuous, strictly
decreasing without bound
(Lemma~\ref{lem:user_customization}), the equation $B(p)+\delta=0$
has a unique solution $\bar p_\delta\geq0$, and
$\mathcal P_\delta=[0,\bar p_\delta]$, with $\bar p_\delta=0$ in the
boundary case $B(0)+\delta=0$. At $p=\bar p_\delta$, we have
$h(\bar p_\delta)=-\delta=B(\bar p_\delta)$, so an accepted default
must attain the maximum of $u_0(\cdot,\bar p_\delta)$, which is
unique; hence $\mathcal D(\bar p_\delta)=\{r_c(\bar p_\delta)\}$.

We next show that $V(p)$ is continuous on $[0,\bar p_\delta]$.
Continuity of $B(p)$ implies continuity of $h(p)$, $K(p)$, and $z(p)$, and
$z(p)\in[-1/e,0)$ by Lemma~\ref{lem:default_acceptance_region}.
Since $m(p)\geq\theta c>0$ and $W_0$ is continuous on $[-1/e,0)$,
the upper endpoint $\overline r(p)$ is continuous. The lower
endpoint is continuous within each regime of
\eqref{eq:acceptance_interval}; it remains to check a junction $p_0$
with $K(p_0)=a$. Approaching $p_0$ within the regime $K(p)<a$
requires $ab>m(p)$, so $ab/m(p_0)\geq1$; then
$z(p_0)=-\bigl(ab/m(p_0)\bigr)e^{-ab/m(p_0)}$ and
$W_{-1}(z(p_0))=-ab/m(p_0)$, giving
$\underline r(p_0)=a/m(p_0)-a/m(p_0)=0$, which matches the regime
$K(p)\geq a$. If instead $ab/m(p_0)<1$, the regime $K(p)<a$ is
infeasible near $p_0$ and no transition occurs. Hence both endpoints
are continuous. Moreover, $r_d\in\mathcal D(p)$ implies
$m(p)r_d\leq ae^{-br_d}+m(p)r_d\leq K(p)$; since
$m(p)\geq\theta c>0$ and $K(p)$ is bounded on the compact interval
$[0,\bar p_\delta]$, all accepted defaults lie in a common bounded
interval. Thus $\mathcal D(\cdot)$ is a continuous, compact-valued
correspondence on $[0,\bar p_\delta]$, and since $G$ is continuous,
Berge's maximum theorem implies that $V(p)$ is continuous; the maximum
in \eqref{eq:best_service_value} is therefore attained.

It remains to compare all provider outcomes. For
$p\in\mathcal P_\delta$: if the user keeps, the payoff is at most
$V(p)\leq V_{\mathrm{serv}}^\star$; if the user customizes, then
$r_c(p)\in\mathcal D(p)$ by the replication argument preceding
Proposition~\ref{prop:fixed_price_default}, so
$G(p,r_c(p))\leq V(p)\leq V_{\mathrm{serv}}^\star$; if the user
exits, the payoff is zero. For $p\notin\mathcal P_\delta$, we have
$B(p)+\delta<0$, and the argument of the first paragraph shows the
user exits, yielding zero. Hence no provider decision exceeds
$\max\{V_{\mathrm{serv}}^\star,0\}$, while $(p^\star,r_d^\star)$
attains $V_{\mathrm{serv}}^\star$ and $\mathcal N$ attains zero.
This proves \eqref{eq:global_stackelberg_value} and the stated
optimality cases.
\end{IEEEproof}

\textit{Tie-breaking robustness:}
Theorem~\ref{thm:stackelberg_characterization} relies on the tie-breaking rule, under which the acceptance region
$\mathcal D(p)$ is closed. Under rejection of a default at indifference, the same provider value can be approached whenever an optimal accepted default can be perturbed into the strict interior of the acceptance set, or when it occurs at $p=\bar p_\delta>0$ and the price can be reduced. When $\delta=0$, customization implements $r_c(p)$ and preserves the provider value. The boundary case $B(0)+\delta=0$ with $\delta>0$ is exceptional: the keep-favoring rule can change the provider value because $p=0$ cannot be reduced.

\section{Experimental Results}\label{Sec:num_result}
In this section, we provide empirical support for the saturating
accuracy--token model, calibrate the service parameters on two
open-weight reasoning models, and illustrate the resulting
equilibrium structure numerically.
\begin{table}[t]
    \centering
    \caption{Average billed non-reasoning tokens ($T_b$) alongside baseline latency ($t_0$) and marginal generation latency per token ($c$).}
    \label{tab:latency_metrics}
    \renewcommand{\arraystretch}{1.2} 
\begin{tabular}{lccccc}
    \toprule
    \multirow{2}{*}{\textbf{Dataset}} & \multirow{2}{*}{\textbf{$T_b$}} & \multicolumn{2}{c}{\textbf{Qwen3-8B}} & \multicolumn{2}{c}{\textbf{DeepSeek-R1-Distill-8B}} \\
    \cmidrule(lr){3-4} \cmidrule(lr){5-6}
    & & $t_0$ (s) & $c$ (s/token) & $t_0$ (s) & $c$ (s/token) \\
    \midrule
    \textbf{AIME 25} & 102 & 0.1775 & 0.000903 & 0.0878 & 0.000679 \\
    \textbf{AIME 24} & 159 & 0.2706 & 0.000760 & 0.1759 & 0.000611 \\
    \textbf{GPQA-D}  & 213 & 0.0899 & 0.001257 & 0.0057 & 0.000356 \\
    \textbf{GSM8K}   & 61  & 0.0663 & 0.000125 & 0.0118 & 0.000128 \\
    \textbf{HMMT 25} & 108 & 0.5198 & 0.001008 & 0.1787 & 0.000738 \\
    \bottomrule
\end{tabular}\end{table}
\begin{table}[t]
\centering
\caption{Fitted accuracy parameters ($D$, $A$, $b$) across datasets.}
\label{tab:parameter_fits}
\resizebox{\columnwidth}{!}{%
\begin{tabular}{@{}lcccccc@{}}
\toprule
\multirow{2}{*}{\textbf{Dataset}} & \multicolumn{3}{c}{\textbf{Qwen3-8B}} & \multicolumn{3}{c}{\textbf{DeepSeek-R1-Distill-8B}} \\
\cmidrule(lr){2-4} \cmidrule(l){5-7}
 & $D$ & $A$ & $b$ & $D$ & $A$ & $b$ \\
\midrule
\textbf{AIME 24} & 0 & 0.595 & $9.25 \times 10^{-5}$ & 0.011 & 0.260 & $1.94 \times 10^{-4}$ \\
\textbf{AIME 25} & 0 & 0.737 & $1.13 \times 10^{-4}$ & 0 & 0.309 & $2.51 \times 10^{-4}$ \\
\textbf{GPQA Diamond}   & 0.354 & 0.227 & $2.63 \times 10^{-4}$ & 0.290 & 0.185 & $2.11 \times 10^{-4}$ \\
\textbf{GSM8K}  & 0.276 & 0.677 & $2.66 \times 10^{-3}$ & 0.067 & 0.822 & $3.03 \times 10^{-3}$ \\
\textbf{HMMT 2025} & 0.044 & 0.302 & $1.89 \times 10^{-4}$ & 0.033 & 0.223 & $1.41 \times 10^{-4}$ \\
\bottomrule
\end{tabular}%
}
\end{table}

\subsection{Experimental Setup}
\label{subsec:num_setup}
We evaluate our framework using two compact open-weight reasoning models: Qwen3-8B \cite{yang2025qwen3technicalreport} and
DeepSeek-R1-Distill-Llama-8B \cite{Guo_2025} (hereafter R1-Distill-Llama-8B) across five
reasoning benchmarks: AIME 2024, AIME 2025, GPQA Diamond, GSM8K,
and HMMT 2025 \cite{rein2024gpqa, cobbe2021trainingverifierssolvemath}. We select open-weight models because they provide
direct control over the inference procedure, allowing us to impose
reasoning-token allocations and measure generation latency consistently
across budget configurations. For GSM8K, we randomly select 500 examples from the
test split using a fixed random seed. For all other benchmarks, we
evaluate the complete available evaluation split. For each question
and reasoning-allocation configuration, we generate three independently
sampled responses. The reported accuracy values are averaged over
both the benchmark questions and the three sampled responses.

We perform inference on models using the vLLM framework~\cite{vllm_paper}
in a Google Colab runtime equipped with a single NVIDIA RTX PRO 6000
Blackwell Server Edition GPU with 96~GB of memory. Thus, both models are served using a single accelerator, illustrating
the feasibility of independently hosted reasoning services based on
compact models. For both models, sampling is performed with temperature
$\tau=0.6$, nucleus-sampling parameter $\mathrm{top}\text{-}p=0.95$,
and $\mathrm{top}\text{-}k=20$. These values follow the recommended
thinking-mode configuration for Qwen3 and are also consistent with
the sampling configuration used in the DeepSeek-R1 reasoning
evaluations. We use the same instruction
that is appended to every benchmark question for both models:
\begin{quote}
\small
\texttt{Please reason step by step, and put your final answer within
\textbackslash boxed\{\}.}
\end{quote}
Additionally, to control inference-time computation, we employ a budget-forcing procedure inspired by the s1 test-time-scaling method~\cite{muennighoff-etal-2025-s1}. For a benchmark containing $N$ questions, the empirical accuracy at
allocation $\ell$ is computed as
\begin{equation}
\widehat{Q}(\ell)
=
\frac{1}{3N}
\sum_{i=1}^{N}
\sum_{j=1}^{3}
\mathbbm{1}
\left\{
\widehat{y}_{i,j}(\ell)=y_i
\right\},
\label{eq:empirical_accuracy}
\end{equation}
where $\widehat{y}_{i,j}(\ell)$ denotes the answer produced by the
$j$th sampled generation for question $i$, and $y_i$ denotes the
corresponding ground-truth answer. The measured accuracy values are
then used to fit the accuracy model in \eqref{eq:accuracy_model}. The resulting measurements and fits
are shown in Fig.~\ref{fig:acc_token_plots} and Table~\ref{tab:parameter_fits}.

We separate the model parameters into two groups according to how they
are obtained. The \emph{service parameters} $T_b$, $t_0$, and $c$ are
taken directly from the measurements reported in
Table~\ref{tab:latency_metrics}, and the
accuracy parameters $D$, $A$, and $b$ are the fitted values from
Table~\ref{tab:parameter_fits}. These characterize the underlying LLM
service and are held at their measured values throughout. The remaining quantities are \emph{economic parameters}: the marginal
token cost $\rho$, the user's accuracy value $v$ and latency
sensitivity $\theta$, the provider's accuracy weight $\alpha$ and
latency cost $\beta$, and the default convenience benefit $\delta$.
Unlike the service parameters above, these are not directly
measurable from model execution; we treat them as modeling inputs and
examine their influence through sensitivity analysis.

\subsection{Model Calibration and Validation}
We now examine whether the measured model behavior supports the
accuracy, token-consumption, and latency models introduced in Section \ref{sec:system_model}. Fig.~\ref{fig:acc_token_plots} shows the empirical accuracy
obtained at different reasoning-allocation configurations together with the
fitted exponential model $Q(r)$ in (\ref{eq:accuracy_model}).

The fitted parameters are reported in Table \ref{tab:parameter_fits}. Across the evaluated model--benchmark pairs, accuracy generally
increases with the reasoning allocation and exhibits diminishing
returns, supporting the saturating form assumed in  $Q(r)$ in  \eqref{eq:accuracy_model}. The fitted parameters also reveal substantial variation across tasks.
A larger value of $A$ indicates greater potential benefit from
additional reasoning, while a larger value of $b$ indicates that these
benefits are realized using fewer reasoning tokens. For example,
GSM8K exhibits relatively rapid saturation, whereas the competition
mathematics benchmarks require substantially larger reasoning
allocations before approaching their fitted accuracy limits. These
differences directly affect equilibrium reasoning allocations
studied in the following subsections. Table~\ref{tab:latency_metrics} reports the average input and output token count
$T_b$ and the latency parameters $t_0$ and $c$. The parameter $t_0$
captures the estimated baseline latency, whereas $c$ represents the
marginal generation latency per reasoning token. The differences
between the two models demonstrate that reasoning allocations with
similar accuracy benefits can nevertheless produce different service
latency and, consequently, different equilibrium decisions.

\subsection{Numerical Results}
\label{subsec:num_results}
Having calibrated the service model, we now examine the economic
equilibrium. The service parameters $T_b,t_0,c,D,A,$ and $b$ are fixed at their calibrated values from
Tables~\ref{tab:latency_metrics}--\ref{tab:parameter_fits}, while the
economic parameters $\rho,v,\theta,\alpha,\beta,$ and $\delta$ are modeling
inputs. Figs.~\ref{fig:numres1} and~\ref{fig:numres3} use the
baseline configuration $v=250$, $\theta=5,$ $\delta=5$,
$\rho=0.005,$ $\alpha=80,$ and $\beta=5$, which places an equilibrium in the interior pricing regime.
Fig.~\ref{fig:numres2} instead uses $v=70,$ $\theta=2,$ $\delta=0,$ $\rho=0.05,$ $\alpha=30,$ and $\beta=10$
with one parameter varied at a time as indicated in the figure.
\begin{figure}[t]
  \centering
\includegraphics[width=0.80\columnwidth]{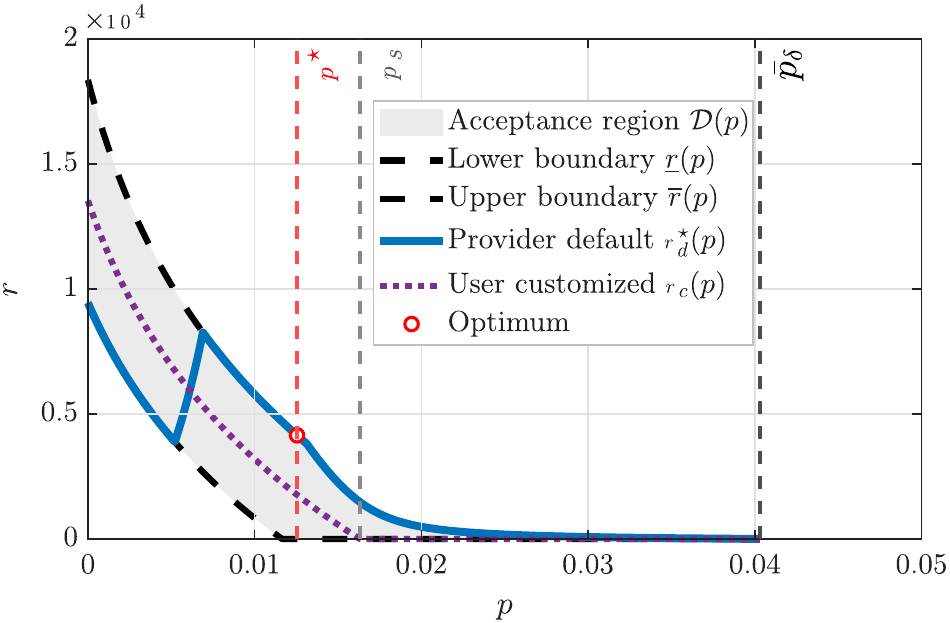}
   \vspace{-0.35cm}
  \caption{Acceptance region $\mathcal D(p)$ and provider-optimal
default $r_d^\dagger(p)$ versus price for AIME~2025 (Qwen3-8B,
baseline configuration).}
  \label{fig:numres1}
  \vspace{-0.35cm}
\end{figure}

\begin{figure}[tb]
  \centering
  \includegraphics[width=0.80\columnwidth]{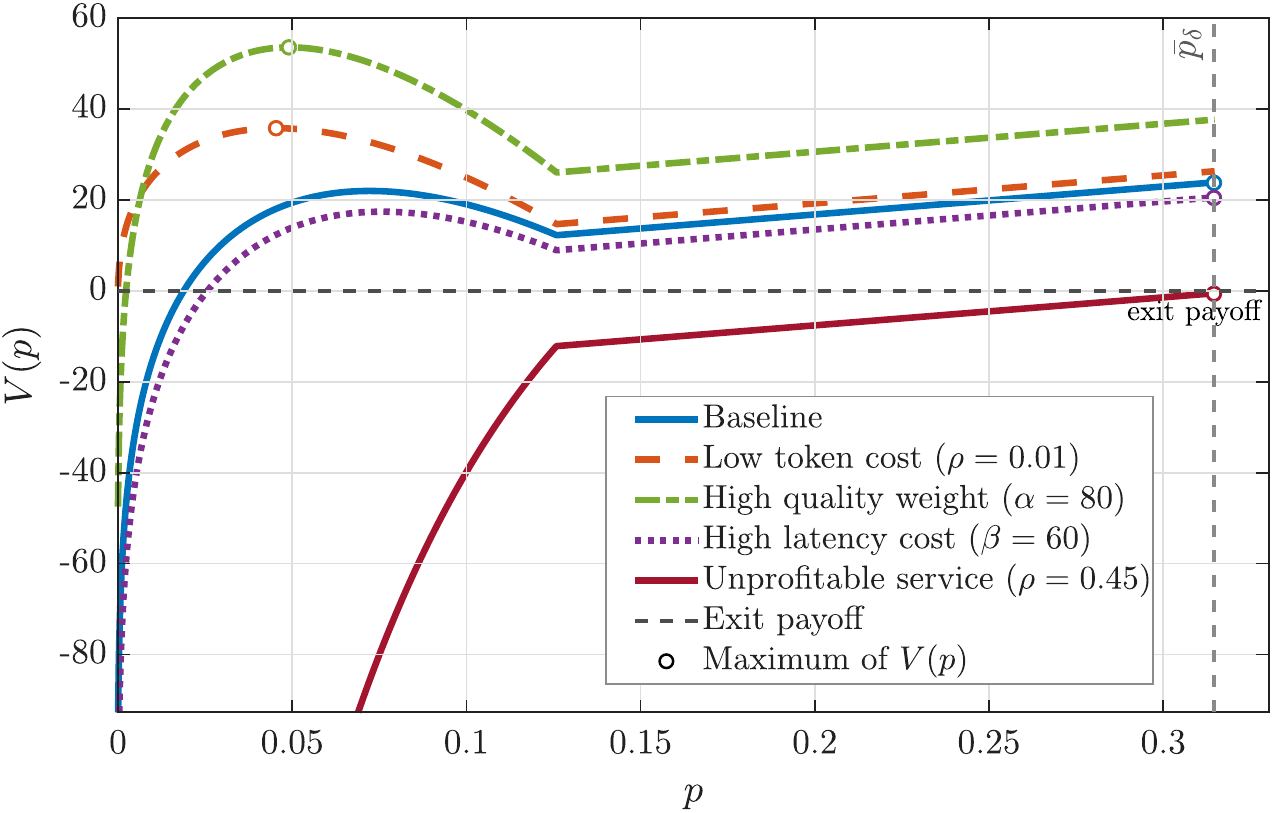}
  \vspace{-0.35cm}
  \caption{Provider value $V(p)$ versus price under alternative
parameter configurations for Qwen3-8B with $\delta=0$.}
  \label{fig:numres2}
  \vspace{-0.35cm}
\end{figure}
Fig.~\ref{fig:numres1} shows the acceptance region $\mathcal D(p)$
and the provider's optimal default policy $r_d^\dagger(p)$ for
AIME~2025. The region narrows as $p$ increases and closes at the
maximum feasible price $\bar p_\delta\approx0.040$. A second
threshold visible in the figure is the \emph{customization shutoff
price} $p_s=vAb-\theta c\approx0.016$, at which
\eqref{eq:positive_customization_condition} holds with equality: for
$p<p_s$, the marginal accuracy value of the first reasoning token
exceeds its marginal cost and the customized allocation $r_c(p)$ is
positive, whereas for $p\geq p_s$, the user selects no additional reasoning tokens under customization, i.e., $r_c(p)=0$, and the customization value reduces to the
corner branch of \eqref{eq:customization_value_closed_form}. The
policy $r_d^\dagger(p)$ traverses the three regimes of
Proposition~\ref{prop:fixed_price_default}: it starts at the lower
boundary, follows the projected interior solution, and terminates at
the upper boundary. An optimal price $p^\star\approx0.0125$, marked
by the circle, lies on the upper-boundary segment, so
$r_d^\star=\overline r(p^\star)$. Since $B(p^\star)>0$, the binding
constraint at acceptance is the comparison with customization,
$U_{\mathcal K}(p^\star,r_d^\star)=B(p^\star)$; that is, the user's
utility loss from the induced allocation relative to customizing at
$r_c(p^\star)$ exactly equals the convenience benefit $\delta$.
Under the tie-breaking rule, the user keeps the default, and the
provider extracts the full convenience margin.

Fig.~\ref{fig:numres2} shows $V(p)$ for GSM8K with $\delta=0$. With
$\delta=0$, we have $h(p)=B(p)$, so a default is accepted only if it
attains the customization value; by uniqueness of the maximizer
(Lemma~\ref{lem:user_customization}),
$\mathcal D(p)=\{r_c(p)\}$ whenever nonempty, and
$V(p)=G(p,r_c(p))$. For $p<p_s$, increasing the price raises the net
per-token margin $\gamma(p)=p-\rho-\beta c$ but reduces the induced
allocation $r_c(p)$; the interplay of these two effects produces the
interior maxima seen in the low-token-cost and high-accuracy-weight
curves. For $p\geq p_s$, we have $r_c(p)=0$ and
$V(p)=(p-\rho)T_b+\alpha D-\beta t_0$, which increases linearly in
$p$ since $T_b>0$, until the participation constraint binds at
$\bar p_\delta$. Accordingly, the low-token-cost ($\rho=0.01$) and
high-accuracy-weight ($\alpha=80$) cases attain interior maxima,
whereas the baseline and high-latency-cost ($\beta=60$) cases are
maximized at the boundary $\bar p_\delta$. For $\rho=0.45$, the
service payoff is negative at every feasible price, so the provider
selects the no-service action $\mathcal N$.

Fig.~\ref{fig:numres3} compares $r_c(p^\star)$ and $r_d^\star$
across the five benchmarks under the baseline configuration. In all
instances, $r_d^\star=\overline r(p^\star)>r_c(p^\star)$: the
provider pushes the default to the upper acceptance boundary, and
the gap $r_d^\star-r_c(p^\star)$ measures the additional reasoning
made acceptable by the convenience benefit $\delta$. For GPQA
Diamond and HMMT~2025, the optimal price exceeds the customization
shutoff price, $p^\star>p_s$, so $r_c(p^\star)=0$ and the positive
defaults are entirely provider-induced. Across benchmarks, the
allocation ordering follows the fitted service parameters: the AIME
benchmarks combine large attainable accuracy gains $A$ with slow
saturation (small $b$), producing the largest allocations; GSM8K
saturates rapidly (large $b$); and GPQA Diamond and HMMT~2025 have
smaller fitted gains, yielding the lowest allocations.

\begin{figure}[tb]
  \centering
\includegraphics[width=0.8\columnwidth]{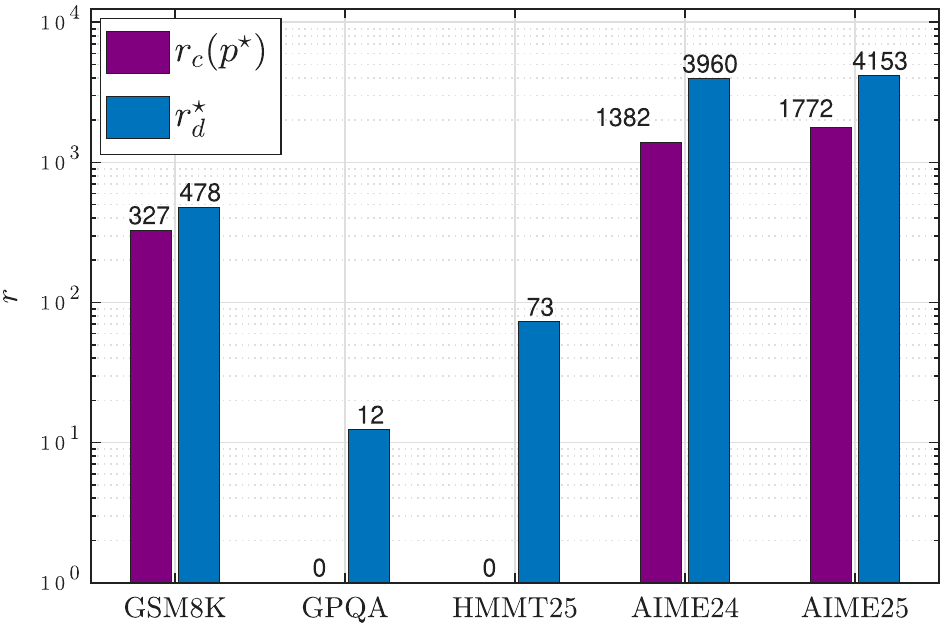}
  \vspace{-0.35cm}
  \caption{Equilibrium reasoning-token allocations across benchmarks
  for Qwen3-8B.}
  \label{fig:numres3}
  \vspace{-0.35cm}
\end{figure}
\subsection{Effect of the Convenience Benefit}
\label{subsec:delta_effect}
\begin{figure}[tb]
  \centering
  \includegraphics[width=0.8\columnwidth]{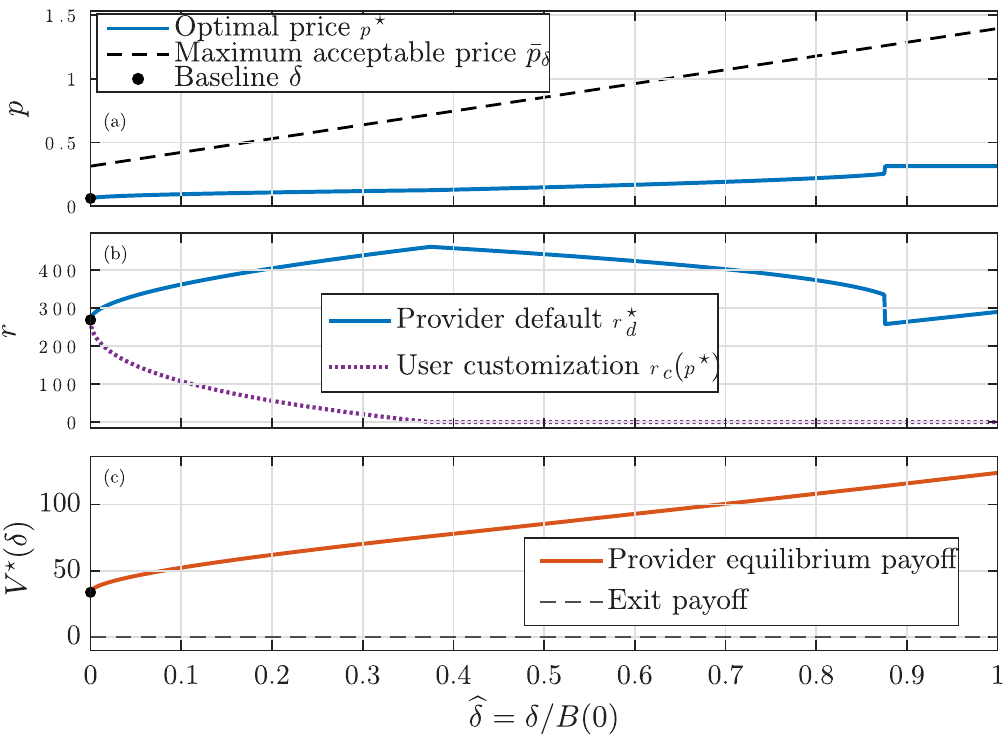}
  \vspace{-0.35cm}
  \caption{Equilibrium price, reasoning allocations, and provider
  payoff versus the normalized convenience benefit
  $\widehat\delta=\delta/B(0)$ for GSM8K (Qwen3-8B, $\alpha=50$).}
  \label{fig:numres4}
  \vspace{-0.35cm}
\end{figure}
Fig.~\ref{fig:numres4} traces an equilibrium against the normalized
convenience benefit $\widehat\delta=\delta/B(0)$ for GSM8K when $B(0)>0$. To
examine a regime in which reasoning is actively induced, we set
$\alpha=50$ and retain the remaining parameters from
Fig.~\ref{fig:numres2}. Fig.~\ref{fig:numres4}(a) shows that
$\bar p_\delta$ increases with $\widehat\delta$, since the
participation condition $B(\bar p_\delta)+\delta=0$ is then
satisfied at a higher price. An equilibrium price $p^\star$ remains
below this cap over most of the range, so the feasible-price
boundary is generally nonbinding. Fig.~\ref{fig:numres4}(b) shows
that $r_c(p^\star)$ decreases to zero as $p^\star$ crosses the
customization shutoff price $p_s$. In contrast, $r_d^\star$ first
increases, then declines, and undergoes a discrete drop near
$\widehat\delta\approx0.88$. At $\widehat\delta=0$, the acceptance
region is the singleton $\{r_c(p)\}$, so the provider cannot steer
the implemented allocation through the default; for
$\widehat\delta>0$, the gap $r_d^\star-r_c(p^\star)$ measures the
additional reasoning made acceptable by the convenience benefit.
Fig.~\ref{fig:numres4}(c) shows that the equilibrium provider payoff
is nondecreasing in $\widehat\delta$. Indeed, increasing $\delta$
weakly enlarges the acceptance region at every price, so every
provider outcome feasible at a smaller $\delta$ remains feasible. Near $\widehat\delta\approx0.88$, the two local maxima of
$V$ exchange global optimality. At the crossing, both branches are
co-optimal, so the equilibrium price correspondence is set-valued at
$\widehat\delta\approx0.88$; the plotted curves select the
global maximizer returned by our grid search, which selects the
price branch at the crossing arbitrarily. The apparent
discontinuity in $p^\star$ and $r_d^\star$ therefore reflects a
switch between co-optimal equilibria rather than a discontinuity in
the equilibrium value, which Fig.~\ref{fig:numres4}(c) confirms is continuous and
nondecreasing.

\section{Conclusion}
\label{sec:conclusion}
In this work, we studied the joint design of token pricing and a default reasoning
allocation in an LLM service, modeling the provider--user
interaction as a Stackelberg game in which the user may keep the
default, customize, or exit. We derived the user's customized
allocation in closed form, characterized the acceptance region
through its Lambert-$W$ boundaries, and reduced the provider's
problem to a one-dimensional price optimization with a three-regime
fixed-price solution, establishing the existence of a Stackelberg
equilibrium and reducing the service-provision decision to a sign comparison of the optimized service value. The analysis isolates the strategic role of the default convenience
benefit. When $\delta=0$, every accepted default coincides with the
user's customized allocation: although pricing and service provision
remain endogenous, the default has no independent allocative power.
When $\delta>0$, the acceptance region contains allocations that
differ from the user's optimum, allowing the provider to steer the
implemented reasoning allocation. Experiments on two open-weight
reasoning models across five benchmarks support the saturating
accuracy--token model and show how model- and task-dependent service
characteristics translate into distinct equilibrium prices,
defaults, and allocations.

Our framework adopts a complete-information, representative-user
model to isolate the default mechanism, abstracting from user and
task heterogeneity, private valuations, and repeated interactions.
Extensions to heterogeneous users, incomplete information, competing
providers, and dynamic pricing are natural directions for future
work, as is the question of how default design and customization
frictions should be governed when they can systematically move
reasoning allocations away from users' independently chosen levels.

\vspace{-0.4em}
\bibliographystyle{IEEEtran}
\bibliography{references}
\end{document}